\documentclass[twoside]{article}
\usepackage[accepted]{aistats2015arxiv}
\usepackage{hyperref}
\usepackage{url}
\usepackage{amssymb}
\usepackage{amsmath}
\usepackage{graphicx}
\usepackage{chngcntr}

\usepackage{url}
\usepackage{graphicx}
\usepackage{amssymb}
\usepackage[round]{natbib}
\usepackage{amsmath}
\usepackage{amsthm}
\usepackage{epstopdf}
\usepackage{lscape}
\usepackage{chngcntr}
\usepackage{rotating}
\newcommand{\be}{\begin{eqnarray} \begin{aligned}}
\newcommand{\ee}{\end{aligned} \end{eqnarray} }
\newcommand{\benn}{\begin{eqnarray*} \begin{aligned}}
\newcommand{\eenn}{\end{aligned} \end{eqnarray*} }
\usepackage{color}

\usepackage{wrapfig}
\usepackage{epsfig}
\usepackage{caption}
\usepackage{subfigure}
\usepackage{hyperref}
\usepackage{algpseudocode}
\usepackage{algorithm}
\newtheorem{theorem}{Theorem}
\newtheorem{corollary}{Corollary}

\usepackage{mathrsfs}

\newcommand{\BEQ}{\begin{equation}}
\newcommand{\EEQ}{\end{equation}}
\newcommand{\BEA}{\begin{eqnarray}}
\newcommand{\EEA}{\end{eqnarray}}
\newcommand{\BGA}{\begin{gather}}
\newcommand{\EGA}{\end{gather}}

\newcommand{\comment}[1]{}

\newcommand{\eps}{\epsilon}
\newcommand{\vx}{{\mathbf x}}

\newcommand{\ii}{{(i)}}

\begin{document}
\setlength{\abovedisplayskip}{0.5mm}
\setlength{\belowdisplayskip}{0.5mm}

\runningtitle{Efficient Estimation of Mutual Information for Strongly Dependent Variables}

\twocolumn[

\aistatstitle{Efficient Estimation of Mutual Information for \\  Strongly Dependent Variables}

\aistatsauthor{ Shuyang Gao \And Greg Ver Steeg \And Aram Galstyan}

\aistatsaddress{  Information Sciences Institute \\ University of Southern California \\ sgao@isi.edu \And Information Sciences Institute \\ University of Southern California \\ gregv@isi.edu 
\And Information Sciences Institute \\ University of Southern California \\ galstyan@isi.edu } ]

\begin{abstract}
We demonstrate that a popular class of non-parametric mutual information (MI) estimators based on $k$-nearest-neighbor graphs requires number of samples that scales exponentially with the true MI. Consequently, accurate estimation of MI between two strongly dependent variables is possible only for prohibitively large sample size. This important yet overlooked shortcoming of the existing estimators is due to their implicit reliance on  local uniformity of the underlying joint distribution. We introduce a new  estimator that is robust to local non-uniformity, works well with limited data, and is able to capture relationship strengths over many orders of magnitude. We demonstrate the superior performance of the proposed estimator on both synthetic and real-world data.
\end{abstract}

\section{Introduction}
\label{sec:intro}
As a measure of the dependence between two random variables, mutual information is remarkably general and has several intuitive interpretations~\citep{cover}, which explains its widespread use in statistics, machine learning, and computational neuroscience; see~\citep{qingwang} for a list of various applications. In a typical scenario, we do not know the underlying joint distribution of the random variables, and instead need to estimate mutual information using i.i.d.\ samples from that distribution. A naive approach to  this problem is to first estimate the underlying probability density, and then  calculate  mutual information using its definition. 
Unfortunately, estimating joint densities from a limited number of samples is often infeasible in many practical settings. %, so that other methods are required.  

A different approach is to estimate mutual information directly from samples in a non-parametric way, without ever describing the entire probability density. The main intuition behind such {\em direct} estimators is that evaluating mutual information can in principle be a more tractable problem than estimating the density over the whole state space~\citep{consistency}. The most popular class of estimators taking this approach is the $k$-nearest-neighbor(kNN) based estimators. One example of such estimator is due to Kraskov, St\"{o}gbauer, and Grassberger (referred to as the KSG estimator from here on)~\citep{kraskov}, which has been extended to generalized nearest neighbors graphs~\citep{GNNG}.
% I'll put a nuanced statement about consistency in Methods
%These types of non-parametric entropy estimates are consistent even though the analogous density estimates are not{\bf Greg will check this}. A related (and provably consistent) estimator based on generalized nearest neighbors graphs was proposed in ~\citep{poczosGNNG}.
%In discrete spaces it is also the case that entropy estimation is much easier than density estimation~\citep{valiant}. 

%Despite these advances, many alternate measures of dependence are used~\citep{hhg,dcor} with new ones being proposed regularly. 
%For instance, Reshef et. al.~\citep{reshef} sought a measure of relationship dependence that would measure the noisiness of the relationship between two variables while being insensitive to the functional relationship.  
%They make the surprising claim that mutual information is inappropriate for this purpose and they back this up with empirical results. 
%Kinney and Atwal~\citep{kinney} point out that the real problem is not that mutual information is the wrong measure, but that their estimates of mutual information (using the KSG estimator) were wildly inaccurate. Kinney and Atwal point out that with more samples the KSG estimator produces the desired results. 

%If the main barrier to using mutual information is difficulty in estimation, the logical next question is whether this difficulty is fundamental or can be overcome through improved estimators. 

Our first contribution is to demonstrate that kNN-based estimators suffer from a critical yet overlooked flaw. Namely, we illustrate that if the true mutual information is $I$, then kNN-based estimators requires exponentially many (in $I$) samples for accurate estimation. This means that \emph{strong relationships} are actually \emph{more difficult} to measure. This counterintuitive property reflects the fact that most work on mutual information estimation has focused on estimators that are good at detecting \emph{independence} of variables rather than precisely measuring strong dependence. In the age of big data, it is often the case that many strong relationships are present in the data and we are interested in picking out only the strongest ones. MI estimation will perform poorly for this task if their accuracy is low in the regime of strong dependence.

We show that the undesired behavior of previous kNN-based MI estimators can be attributed to the assumption of local uniformity utilized by those estimators, which can be violated  for sufficiently strong (almost deterministic) dependencies. As our second major contribution, we suggest a new kNN estimator that relaxes the local uniformity condition by introducing a correction term for local non-uniformity. We demonstrate empirically that for strong relationships, the proposed estimator needs significantly fewer samples for accurately estimating mutual information.

In Sec.~\ref{sec:methods}, we introduce kNN-based non-parametric entropy and mutual information estimators. In Sec.~\ref{sec:lim}, we demonstrate the limitations of kNN-based mutual information estimators. And then we suggest a correction term to overcome these limitations in Sec.~\ref{sec:corr}. In Sec.~\ref{sec:results}, we show empirically using synthetic and real-world data that our method outperforms existing techniques. In particular, we are able to accurately measure strong relationships using smaller sample sizes. Finally, we conclude with related work and discussion.

\section{kNN-based Estimation of Entropic Measures}\label{sec:methods}

In this section, we will first introduce the naive kNN estimator for mutual information, which is based on an entropy estimator due to~\citep{kNN_naive}, and show its theoretical properties. Next, we focus on a popular variant of kNN estimators, called KSG estimator~\citep{kraskov}.

\subsection{Basic Definitions}
\label{sec:problem}

Let $\vx=(x_1,x_2,...,x_d)$ denote a $d$-dimensional absolute continuous random variable whose probability density function is defined as $p: \mathbb R^{d} \to \mathbb R$ and marginal densities of each $x_j$ are defined as $p_j: \mathbb R \to \mathbb R, j=1,\ldots,d$. Shannon Differential Entropy and Mutual Information are defined in the usual way:
\begin{eqnarray}
H\left( \vx \right) &=& -\int\limits_{{{\mathbb R}^d}} {p\left( \vx \right)\log p\left( \vx \right)d\vx}\\
I\left( \vx \right) &=& \int\limits_{{{\mathbb R}^d}} {p\left( \vx \right)\log \frac{{p\left( \vx \right)}}{{\prod\limits_{j = 1}^d {{p_j}\left( {{x_j}} \right)} }}d\vx} 
\end{eqnarray}
%\be
%I\left( \vx \right) = \int\limits_{{{\mathbb R}^d}} {p\left( \vx \right)\log \frac{{p\left( \vx \right)}}{{\prod\limits_{j = 1}^d {{p_j}\left( {{x_j}} \right)} }}d\vx} 
%\ee
We use natural logarithms so that information is measured in nats. For $d>2$, the generalized mutual information is also called \textit{total correlation}~\citep{watanabe1960} or \textit{multi-information}~\citep{studeny1998}.

Given  $N$ i.i.d.\ samples, $\mathcal{X}=\left\{ \vx^\ii \right\}_{i = 1}^n$, drawn from $p(\vx)$, the task is to construct a mutual information estimator $\hat I\left( \vx \right)$ based on these samples. %As we will show below, standard non-parametric  estimators based on kNN require at least $N = O( e^{C\times I( \vx )})$ samples to accurately estimate $\hat I\left( \vx \right)$ where $C$ is a constant based on the dimensionality $d$.  

\subsection{Naive kNN Estimator}
\paragraph{Entropy Estimation}
The naive kNN entropy estimator is as follows:
\be\label{eq:kNN_h_old}
{\widehat H'_{kNN,k}}\left( {\bf{x}} \right) =  - \frac{1}{n}\sum\limits_{i = 1}^n {\log {{\widehat p}_k}\left( {{{\bf{x}}^\ii}} \right)}
\ee
where 
\be\label{eq:kNN_h_old_p}
{\widehat p_k}\left( {{{\bf{x}}^\ii}} \right) = \frac{k}{{n - 1}}\frac{{\Gamma \left( {d/2 + 1} \right)}}{{{\pi ^{d/2}}}}{r_k}{\left( {{{\bf{x}}^\ii}} \right)^{ - d}}
\ee
${r_k}\left( {{{\bf{x}}^\ii}} \right)$ in Eq.~\ref{eq:kNN_h_old_p} is the Euclidean distance from $\vx^\ii$ to its $k$th nearest neighbor in $\mathcal{X}$. By introducing a correction term, an asymptotic unbiased estimator is obtained:
\be\label{eq:kNN_h}
{\widehat H_{kNN,k}}\left( {\bf{x}} \right) =  - \frac{1}{n}\sum\limits_{i = 1}^n {\log {{\widehat p}_k}\left( {{{\bf{x}}^\ii}} \right)} - \gamma_k
\ee
where 
\BEQ
{\gamma _k} =  \frac{{{k^k}}}{{\left( {k - 1} \right)!}}\int_0^\infty  {\log \left( x \right){x^{k - 1}}{e^{ - kx}}dx} = \psi(k)-\log(k) \\
\EEQ
where $\psi(\cdot)$ represents the digamma function. 

The following theorem shows asymptotic unbiasedness of ${\widehat H_{kNN,k}}\left( {\bf{x}} \right)$ according to~\citep{kNN_naive}. 
\begin{theorem}[kNN entropy estimator, asymptotic unbiasedness,~\citep{kNN_naive}]\label{theo:kNN_ent}
Assume that $\vx$ is absolutely continuous and $k$ is a positive integer, then 
\be
{\lim _{n \to \infty }}\mathbb{E}\left[ {{{\widehat H}_{kNN,k}}\left( {\mathbf{x}} \right)} \right] = {H_{kNN,k}(\vx)}
\ee
i.e., this entropy estimator is asymptotically unbiased.
\end{theorem}

\paragraph{From Entropy to Mutual Information} To construct a mutual information estimator from an entropy estimator is straightforward by combining entropy estimators using the identity~\citep{cover}:
\be\label{eq:mie}
I(\vx) = \sum\limits_{i = 1}^d {H\left( {{x_i}} \right)}  - H\left( \vx \right)
\ee
Combining Eqs.~\ref{eq:kNN_h_old} and~\ref{eq:mie}, we have,
\BEQ\label{eq:kNN_mi}
{\widehat I'_{kNN,k}}\left( {\mathbf{x}} \right) = \frac{1}{n}\sum\limits_{i = 1}^n {\log \frac{{{{\widehat p}_k}\left( {{{\mathbf{x}}^\ii}} \right)}}{{{{\widehat p}_k}\left( {{{\mathbf{x}}^\ii_1}} \right){{\widehat p}_k}\left( {{{\mathbf{x}}^\ii_2}} \right)...{{\widehat p}_k}\left( {{{\mathbf{x}}^\ii_d}} \right)}}} \\
\EEQ
where $\vx^\ii_j$ denotes the point projected into $j$th dimension in $\vx^\ii$ and ${\widehat p}_k(\vx^\ii_j)$ represents the marginal kNN density estimator projected into $j$th dimension of $\vx^\ii$.

Similar to Eq.~\ref{eq:kNN_h}, we can also construct an asymptotically unbiased mutual information estimator based on Theorem~\ref{theo:kNN_ent}:
\be
{\widehat I_{kNN,k}} = {\widehat I'_{kNN,k}}-\left( {d - 1} \right){\gamma _k}
\ee
\begin{corollary}[kNN MI estimator, asymptotic unbiasedness] 
Assume that $\vx$ is absolute continuous and $k$ is a positive integer, then:
\be
{\lim _{n \to \infty }}\mathbb{E}\left[ {{{\widehat I}_{kNN,k}}\left( {\mathbf{x}} \right)} \right] = {I_{kNN,k}}\left( {\mathbf{x}} \right)
\ee
\end{corollary}
\subsection{KSG Estimator}
The KSG mutual information estimator~\citep{kraskov} is a popular variant of the naive kNN estimator. The general principle of KSG is that for each density estimator in different spaces, we would like to use the similar length-scales for $k$-nearest-neighbor distance as in the joint space so that  the bias would be approximately smaller. Although the theoretical properties of this estimator is unknown, it has a relatively good performance in practice, see~\citep{khan} for a comparison of different estimation methods. Unlike naive kNN estimator, the KSG estimator uses the max-norm distance instead of L2-norm. In particular, if $\eps_{i,k}$ is twice the (max-norm) distance to the $k$-th nearest neighbor of $\vx^\ii$, it can be shown that the expectation value (over all ways of drawing the surrounding $N-1$ points) of the log probability mass within the box centered at $\vx^\ii$ is given by this expression.  
\BEA
\mathbb E_{\eps_{i,k}}\hspace{-1mm} \left[\log \int_{|\vx-\vx^\ii|_\infty \leq \eps_{i,k}/2} \hspace{-1mm} f(\vx)~d\vx  \right] = \psi(k)-\psi(N) \nonumber \\
\EEA

We use $\psi$ to represent the digamma function. 
If we assume that the density inside the box (with sides of length $\eps_{i,k}$) is constant, then the integral becomes trivial and we find that, 
\be
\log (f(\vx_i) \eps_{i,k}^{d}) = \psi(k)-\psi(N).
\ee
Rearranging and taking the mean over $-\log f(\vx^\ii)$ leads us to the following entropy estimator: 
 \BEA\label{eq:kl}
   \hat H_{KSG,k}(\vx) &\equiv& \psi(N)-\psi(k)+ \frac{d}{N} \sum_{i=1}^N \log \epsilon_{i,k} \nonumber \\   
 \EEA

Note that $k$, defining the size of neighborhood to use in local density estimation, is a free parameter. Using smaller $k$ should be more accurate, but larger $k$ reduces the variance of the estimate~\citep{khan}. While consistent density estimation requires $k$ to grow with $N$~\citep{kNNdensity}, entropy estimates converge almost surely for any fixed $k$~\citep{qingwang, consistency} under some weak conditions.\footnote{This assumes the probability density is absolutely continuous but see \citep{GNNG} for some technical concerns.} 

To estimate the mutual information, in the joint $\vx$ space we set $k$, the size of the neighborhood, which determines $\eps_{i,k}$ for each point $\vx^\ii$. Next, we consider the smallest rectilinear hyper-rectangle that contains these $k$ points, which has sides of length $\eps^{x_j}_{i,k}$ for each marginal direction $x_j$. We refer to this as the ``max-norm rectangle'' (as shown in Fig.~\ref{fig:ksg}).
Let $n_{x_j}$ be the number of points at a distance less than or equal to $\eps^{x_j}_{i,k}/2$ in the $x_j$-subspace.
For each marginal entropy estimate, we use $n_{x_j}(i)$ instead of $k$ to set the neighborhood size at each point. 
Finally, in the joint space using a rectangle instead of a box in Eq.~\ref{eq:kl} leads to a correction term of size $(d-1)/k$ (details given in~\citep{kraskov}). 
 Adding the entropy estimators together with these choices yields the following.
 \BEA\label{eq:kraskov}
   \hat I_{KSG,k}(\vx) \equiv
  (d-1)\psi(N)+ \psi(k) - (d-1)/k \nonumber
   \\- \frac{1}{N} \sum_{i=1}^N \sum_{j=1}^d \psi (n_{x_j}(i))
 \EEA
% Note that a small correction has been dropped in this expression, see \citep{kraskov} for details. 

\begin{figure}[htbp] %  figure placement: here, top, bottom, or page
   \centering
   \subfigure[]{\includegraphics[width=0.2\textwidth]{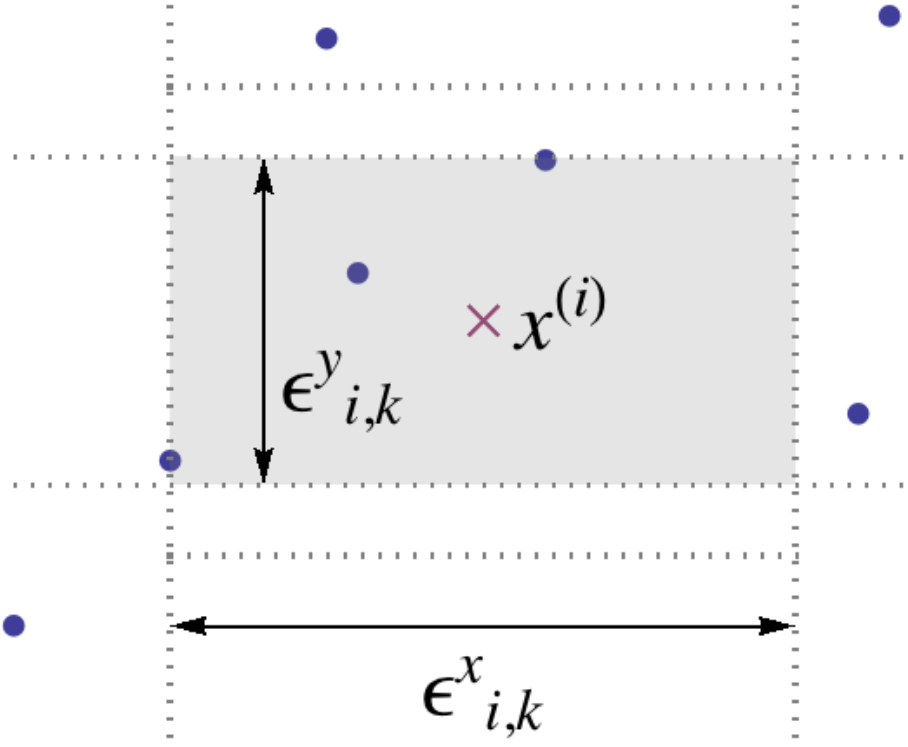} \label{fig:ksg}} 
   \quad
   \subfigure[]{ \raisebox{0.22\height}{\includegraphics[width=0.2\textwidth]{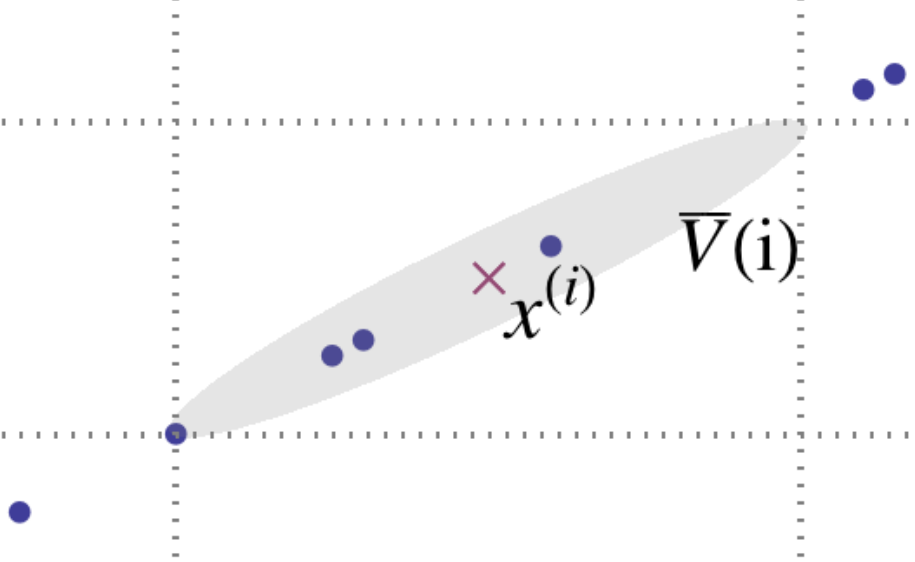} \label{fig:ksg_problem}} } 
   \caption{Centered at a given sample point, $\vx^\ii$, we show the max-norm rectangle containing $k$ nearest neighbors (a) for points drawn from a uniform distribution, $k=3$, and (b) for points drawn from a strongly correlated distribution, $k=4$.}  
\end{figure}

\section{Limitations of kNN-based MI Estimators}
\label{sec:lim}
In this section, we demonstrate a significant flaw in kNN-based MI estimators which is summarized in the following theorems. We then explain why these estimators fail to accurately estimate mutual information unless the correlations are relatively weak or the number of samples is large. 

\begin{theorem} \label{theo:kNN}
For any d-dimensional absolute continuous probability density function, $p(\vx)$, for any $k \geq 1$, for the estimated mutual information to be close to the true mutual information, $| \hat I_{kNN,k}(\vx) - I(\vx) | \leq \varepsilon$, requires that the number of samples, $N$, is at least, $N \ge C \exp \left( {\frac{{I\left( {\mathbf{x}} \right) - \varepsilon }}{{d - 1}}} \right) + 1$, where $C$ is a constant which scales like $O(\frac{1}{d})$.
\end{theorem}
The proof of Theorem~\ref{theo:kNN} is shown in the Appendix~\ref{sec:derive_kNN}.
\begin{theorem}\label{theo:ksg}
For any d-dimensional absolute continuous probability density function, $p(\vx)$, for any $k\geq 1$, for the estimated mutual information to be close to the true mutual information, $| \hat I_{KSG,k}(\vx) - I(\vx) | \leq \varepsilon$, requires that the number of samples, $N$, is at least, $N \ge C\exp \left( {\frac{{I\left( {\mathbf{x}} \right) - \varepsilon }}{{d - 1}}} \right)+1$, where $C={{\mathbf{e}}^{ - \frac{{k - 1}}{k}}}$. 
\end{theorem}
\begin{proof}
Note that $\psi(n) = H_{n-1} - \gamma$, where $H_{n}$ is the $n$-th harmonic number and  $\gamma \approx 0.577$ is the Euler-Mascheroni constant.
\begin{eqnarray}
\widehat{I}_{KSG,k}(\vx)  &\le&  (d-1)\psi(N)+\psi(k) - (d-1)/k  \nonumber \\
&~&- \frac{1}{N}\sum\limits_{i = 1}^N {\sum\limits_{j = 1}^d {\psi \left( k \right)} }  \nonumber \\
%&\le&  (d-1)\psi(N)+\psi(k)-(d-1)/k-d  \psi(k)  \nonumber \\
&=&(d-1)(\psi(N)-\psi(k)-1/k)  \nonumber \\
&=&(d-1)(H_{N-1}-\gamma - H_{k-1} + \gamma - 1/k) \nonumber  \\
&\le&(d-1)(\log(N-1) + (k-1)/k) \nonumber \\
\end{eqnarray}
The first inequality is obtained from Eq.~\ref{eq:kraskov} by observing that ${n_{{x_j}}}\left( i \right) \ge k$ for any $i,j$ and that $\psi(k)$ is a monotonically increasing function.  And the last inequality is obtained by dropping the term $-H_{k-1}\le 0$ and using the well-known upper bound $H_N \le \log N + 1$. Requiring that $ |\hat I_{KSG,k}(\vx)   -  I(\vx)| < \varepsilon$, we obtain $N \ge C \exp \left( {\frac{{I\left( {\mathbf{x}} \right) - \varepsilon}}{{d - 1}}} \right)+1$, where $C={{{e}}^{ - \frac{{k - 1}}{k}}}.$
\end{proof}

The above theorems state that for any fixed dimensionality, the number of samples needed for estimating mutual information $I(\mathbf{x})$ increases exponentially with the {\em magnitude} of  $I(\mathbf{x})$. From the point of view of determining independence, i.e., distinguishing $I(\vx) =0$ from $I(\vx) \neq 0$, this restriction is not particularly troubling. However, for finding strong signals in data it presents a major barrier. Indeed, consider two random variables $X$ and $Y$, where $X\sim \mathcal{U}(0,1)$ and $Y=X+\eta \mathcal{U}(0,1)$. When $\eta \rightarrow 0$, the relationship between $X$ and $Y$ becomes nearly functional, and the mutual information diverges as $I(X:Y)\rightarrow \log\frac{1}{\eta}$. As a consequence, the number of samples needed for accurately estimating $I(X:Y)$ diverges as well. This is depicted in Fig.~\ref{fig:simple_model_lower_bound} where we compare  the empirical lower bound to our theoretical bound given by Theorem~\ref{theo:ksg}. It can be seen that the theoretical bounds are rather conservative, but they have the same exponentially growing rates comparing to the empirical ones.

%To develop an intuitive understanding  of this undesired behavior, we recall that kNN-based estimators assume local uniformity of the underlying joint distribution. For the above example, this assumption can potentially break down when the relationship between the variables become nearly functional; see Fig.~\ref{}. Indeed

What is the origin of this undesired behavior? An intuitive and general argument comes from looking at the assumption of local uniformity in kNN-based estimators. In particular, both naive kNN and KSG estimators approximate the probability density in the kNN ball or max-norm rectangle containing the $k$ nearest neighbors with uniform density. If there are strong relationships (in the joint $\vx$ space, the density becomes more singular), then we can see in Fig.~\ref{fig:ksg_problem} that the uniform assumption becomes problematic. 

%(LNC) that relaxes the local uniformity assumption, by adding a local non uniformity correction (LNC) term.  with a local linearity assumption, which results in a m. Instead, the propose estimator relies on local lin and demonstrate empirically that this modification results in a significantly more efficient estimator

%From the point of view of determining independence, i.e., distinguishing $I(\vx) =0$ from $I(\vx) \neq 0$, this restriction is not particularly troubling. However, for finding strong signals in data it presents a major barrier. A recent nearest-neighbor graph based estimator~\citep{GNNG} also has an upper bound that scales like $\log N$ for estimating (Shannon) mutual information.

\begin{figure}[ht]
\centering
\includegraphics[width=0.8\linewidth]{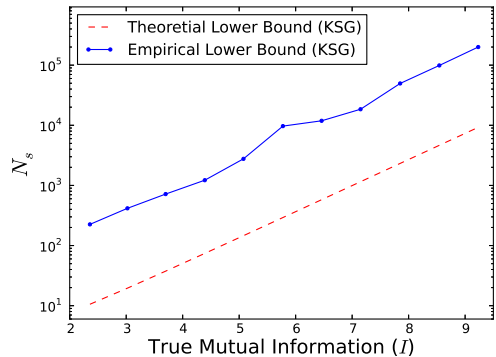}
\caption{A semi-logarithmic plot of $N_s$ (number of required samples to achieve an error at most $\varepsilon$) for KSG estimator for different values of $I(X:Y)$. We set $\varepsilon = 0.1$, $k=1$. 
}

\label{fig:simple_model_lower_bound}
\end{figure}

\section{Improved kNN-based Estimators}
\label{sec:corr}
In this section, we suggest a class of kNN-based estimators that relaxes the local uniformity assumption. 
\paragraph{Local Nonuniformity Correction (LNC)}
Our second contribution is to provide a general method for overcoming the limitations above. 
%We will focus on the KSG estimator since it is used more widely in practice.
Considering the ball (in naive kNN estimator) or max-norm hyper-rectangle (in KSG estimator) around the point $\vx^\ii$ which contains $k$ nearest neighbors, let us denote this region of the space with $\mathcal V(i) \subset \mathbb R^{d}$, whose volume is $V(i)$. 
Instead of assuming that the density is uniform inside $\mathcal V(i)$ around the point $\vx^\ii$, we assume that there is some subset,
$\bar{\mathcal V}(i) \subseteq \mathcal V(i)$ 
with volume 
$\bar V(i) \leq V(i)$ 
on which the density is constant, i.e., 
$\hat p (\vx^\ii) = \frac{\mathbb I[\vx \in \bar {\mathcal V}(i)]}{\bar V(i)}.$
This is illustrated with a shaded region in Fig.~\ref{fig:ksg_problem}. 
We now repeat the derivation above using this altered assumption about the local density around each point for $\hat H(\vx)$. 
We make no changes to the entropy estimates in the marginal subspaces. 
Based on this idea, we get a general correction term for kNN-based MI estimators (see Appendix~\ref{sec:derive_lnc} for the details of derivation):

\be\label{eq:lnc}
\hat I_{LNC}(\vx) = \hat I(\vx) - \frac1N \sum_{i=1}^N \log \frac{\bar V(i)}{V(i)}
\ee
where $\hat I(\vx)$ can be either $\hat I_{kNN}(\vx)$ or $\hat I_{KSG}(\vx)$.

%\shuyang{The details of the derivation are provided in the appendix.}
%\\

If the  local density in the $k$-nearest-neighbor region $\mathcal V(i)$ is highly non-uniform, as is the case for strongly related variables like those in Fig.~\ref{fig:ksg_problem}, then the proposed correction term will improve the estimate. For instance, if we assume that relationships in data are smooth, functional relationships plus some noise, the correction term will yield significant improvement as demonstrated empirically in Sec.~\ref{sec:results}. We note that this correction term is not bounded by $N$, but rather by our method of estimating $\bar V$.  Next, we will give one concrete implementation of this idea by focusing on the modification of KSG estimator.

\paragraph{Estimating Nonuniformity by Local PCA}
With the correction term in Eq.~\ref{eq:lnc}, we have transformed the problem into that of finding a local volume on which we believe the density is positive. Regarding KSG estimator, instead of a uniform distribution within the max-norm rectangle in the neighborhood around the point $\vx^\ii$, we look for a small, rotated (hyper)rectangle that covers the neighborhood of $\vx^\ii$. 
%\note{here we put a figure to illustrate the idea.} 
The volume of the rotated rectangle is obtained by doing a localized principle component analysis around all of $\vx^\ii$'s $k$ nearest neighbors, and then multiplying the maximal axis values together in each principle component after the $k$ points are transformed to the new coordinate system\footnote{Note that we manually set the mean of these $k$ points to be $\vx^\ii$ when doing PCA, in order to put $\vx^\ii$ in the center of the rotated rectangle.}. The key advantage of our proposed estimator is as follows: while KSG assumes \textit{local uniformity} of the density over a region containing $k$ nearest neighbors of a particular point, our estimator relies on a much weaker assumption of local linearity over the same region. Note that the \textit{local linearity} assumption has also been widely adopted in the manifold learning, for example, \textit{local linear embedding}(LLE)~\citep{LLE} and \textit{local tangent space alignment}(LTSA)~\citep{LTSA}.

\paragraph{Testing for Local Nonuniformity}
One problem with this procedure is that we may find that, locally, points may occupy a small sub-volume, 
%A problem may be caused by PCA is that, it always tends to learn a linear structure from limited number of neighbors. 
i.e., even if the local neighborhood is actually drawn from a uniform distribution(as shown in Fig.~\ref{fig:ksg}), the volume of the PCA-aligned rectangle will with high probability be smaller than the volume of the max-norm rectangle, leading to an artificially large non-uniformity correction.

To avoid this artifact, we consider a trade-off between the two possibilities: for a fixed dimension $d$ and nearest neighbor parameter $k$, we find a constant ${\alpha _{k,d}}$, such that if $\bar V(i) / V(i) < \alpha_{k,d}$, then we assume local uniformity is violated and use the correction $\bar V(i)$, otherwise the correction is discarded for point $\vx^\ii$. Note that if $\alpha_{k,d}$ is sufficiently small, then the correction term will be always discarded, so that our estimator reduces to the KSG estimator. Good choices of $\alpha_{k,d}$ are set using arguments described in Appendix~\ref{sec:derive_alpha}. Furthermore, we believe that as long as the expected value of $\mathbb E[\bar V(i) / V(i)] \ge \alpha_{k,d}$ in large $N$ limit, for some properly selected $\alpha_{k,d}$, then the consistency properties of the proposed estimator will be identical to the constancy properties of the KSG estimator. A rigorous proof of this hypothesis is left as a future work.

The full algorithm for our estimator is given in Algorithm~\ref{alg:mie}.
\begin{algorithm}[h]
\caption{\textbf{Mutual Information Estimation with Local Nonuniform Correction}}
\label{alg:mie}
\begin{algorithmic}

\State{\textbf{Input:} points $\vx^{(1)},\vx^{(2)},...,\vx^{(N)}$, parameter $d$ (dimension), $k$ (nearest neighbor), $\alpha_{k,d}$}
\State{\textbf{Output:} $\hat I_{LNC}(\vx)$}

\State{Calculate $\hat I_{KSG}(\vx)$ by KSG estimator, using the same nearest neighbor parameter $k$}
\For{each point $\vx^\ii$}
\State{Find $k$ nearest neighbors of $\vx^\ii$: $kNN^\ii_1$, $kNN^\ii_2$,..., $kNN^\ii_k$} 
\State{Do PCA on these $k$ neighbors, calculate the volume corrected rectangle $
\bar V(i)$}
\State{Calculate the volume of max-norm rectangle $V(i)$}
\If{$\bar V(i) / V(i) < \alpha_{k,d}$}
\State{$LNC_i = \log \frac{\bar V(i)}{V(i)}$}
\Else
\State{$LNC_i = 0.0$}
\EndIf
\EndFor

\State{Calculate $LNC^*$: average value of $LNC_1,LNC_2,...,LNC_N$}
\State{$\hat I_{LNC} = \hat I_{KSG} - LNC^*$}

\end{algorithmic}
\end{algorithm}

\section{Experiments}
\label{sec:results}

We evaluate the proposed estimator on both synthetically generated and real-world data. For the former, we considered various functional relationships and thoroughly examined the performance of the estimator over a range of noise intensities. For the latter, we applied our estimator to the WHO dataset used previously in~\citep{reshef}. Below we report the results.

\subsection{Experiments with synthetic data }
\paragraph{Functional relationships in two dimensions} 
In the first set of experiments, we generate samples from various functional relationships of the form $Y=f(X)+\eta$ that were previously studied in~\citet{reshef, kinney}. The noise term $\eta$ is distributed uniformly over the interval $\left[ { - \sigma/2, \sigma/2 } \right]$, where  $\sigma$ is used to control the noise intensity. We also compare the results to several baseline estimators: KSG~\citep{kraskov}, generalized nearest neighbor graph (GNN)~\citep{GNNG}~\footnote{We use the online code \url{http://www.cs.cmu.edu/~bapoczos/codes/REGO\_with\_kNN.zip} for the GNN estimator.}, minimum spanning trees (MST)~\citep{spanning, MST_MI}, and exponential family with maximum likelihood estimation (EXP)~\citep{EXP_MI1}~\footnote{We use the Information Theoretical Estimators Toolbox (ITE)~\citep{ITEtoolbox} for MST and EXP estimators.}. 
\begin{figure}[htbp]
	\centering
	\includegraphics[width=1.0\linewidth]{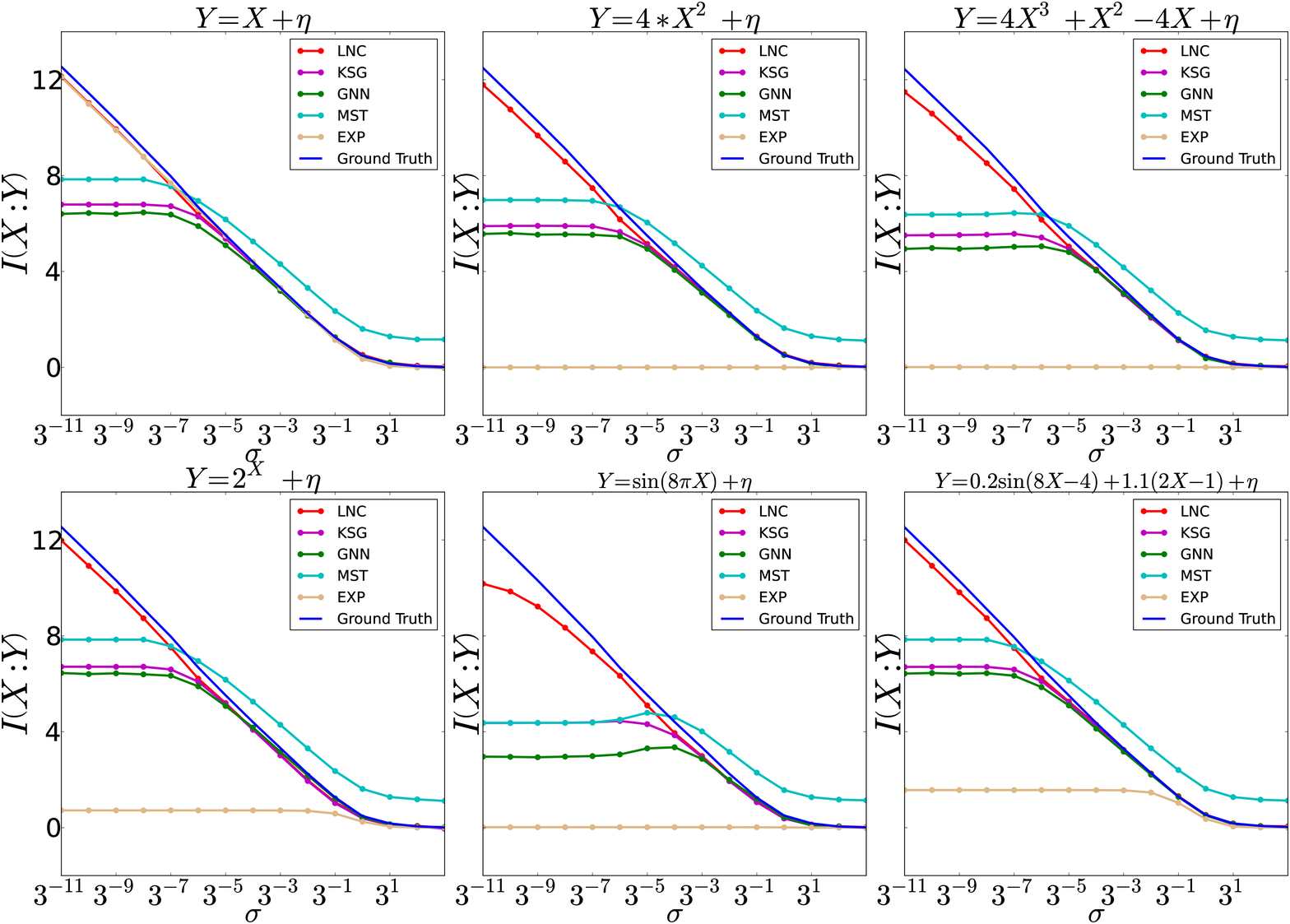}
	\caption{For all the functional relationships above, we used a sample size $N=5000$ for each noise level and the nearest neighbor parameter $k=5$ for LNC, KSG and GNN estimators.}
	\label{fig:synthetic_2d}
\end{figure}

Figure~\ref{fig:synthetic_2d} demonstrates that the proposed  estimator, LNC, consistently outperforms the other estimators. Its superiority is most significant for the low noise regime. In that case, both KSG and GNN estimators are bounded by the sample size while LNC keeps growing. MST tends to overestimate MI for large noise but still stops growing after the noise fall below a certain intensity\footnote{Even if $k=2$ or $3$, KSG and GNN estimators still stop growing after the noise fall below a certain intensity.}. Surprisingly, EXP is the only other estimator that performs comparably with LNC for the linear relationship (the left most plot in Figure~\ref{fig:synthetic_2d}). However, it fails dramatically for all the other relationship types. See Appendix~\ref{sec:more} for more functional relationship tests.
\paragraph{Convergence rate} Figure~\ref{fig:converge_rate} shows the convergence of the two estimators $\hat I_{KSG}$ and $\hat I_{LNC}$, at a fixed (small) noise intensity, as we vary the sample size. We test the estimators on both two and five dimensional data for linear and quadratic relationships. 
%For the 2D dataset, we set $k$ to $5$ and $\alpha_{k,d}$ to $0.37$, while for the 5D dataset we  set $k$ to $8$ and $\alpha_{k,d}$ to $0.12$;  the procedure for selecting $\alpha_{k,d}$ is described in the appendix. 
We observe that LNC is doing better overall. In particular, for linear relationships in 2D and 5D, as well as quadratic relationships in 2D, the required sample size for LNC is several orders of magnitude less that for $\hat I_{KSG}$.  For instance, for the 5D linear relationship $KSG$ does not converge even for the sample size  ($10^5$) while $LNC$ converges to the true value with only $100$ samples.  

Finally, it is worthwhile to remark on the relatively slow convergence of  $LNC$ for the 5D quadratic example. This is because $LNC$, while relaxing the local uniformity assumptions, still assumes local linearity. The stronger the nonlinearity, the more samples are required to find neighborhoods that are locally approximately linear. We note, however, that LNC still converges faster than KSG.

\begin{figure}[ht]
\centering
\minipage{0.22\textwidth}
\includegraphics[width=\linewidth]{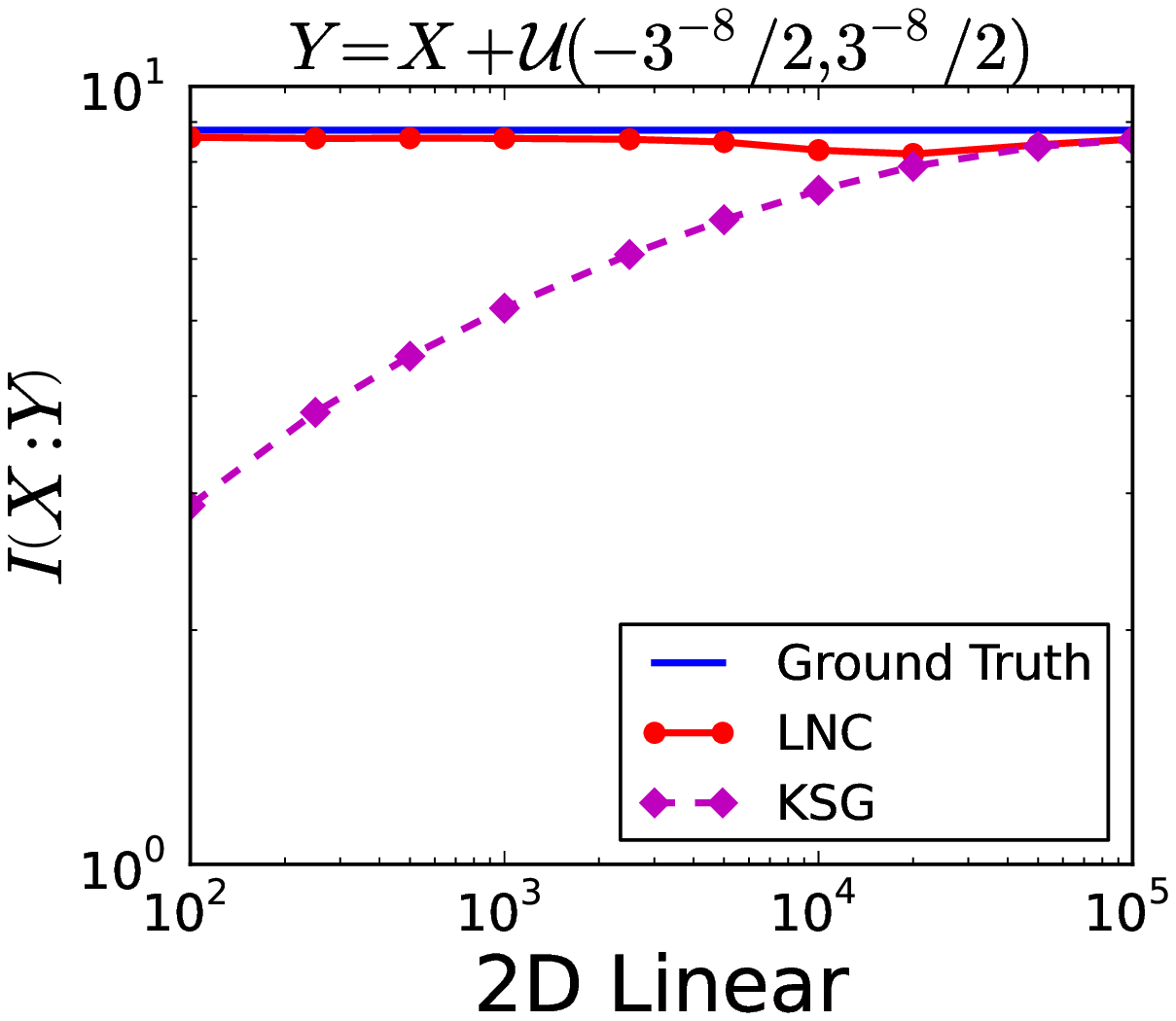}
\endminipage
\minipage{0.22\textwidth}
\includegraphics[width=\linewidth]{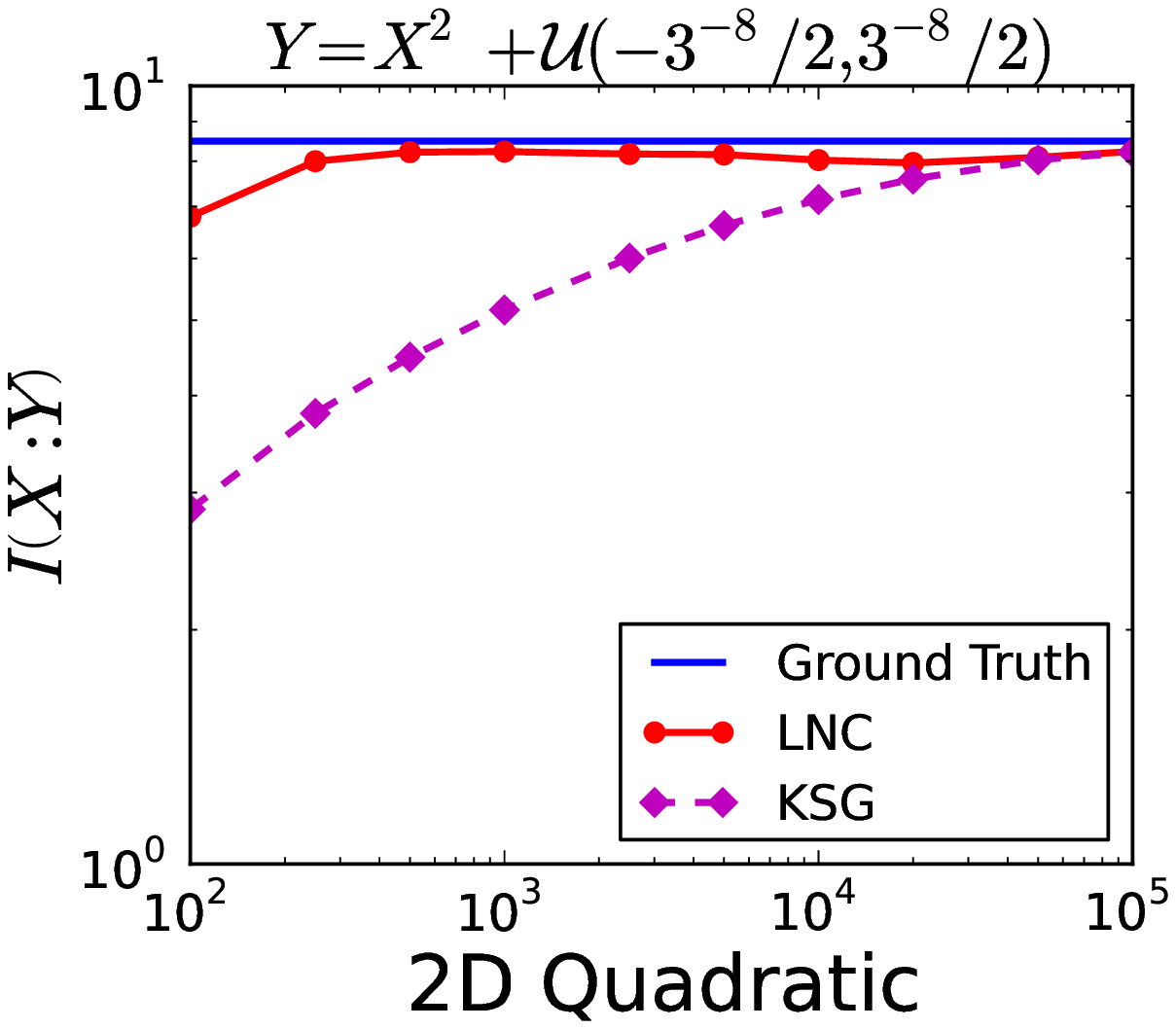}
\endminipage
\\

\minipage{0.22\textwidth}
\includegraphics[width=\linewidth]{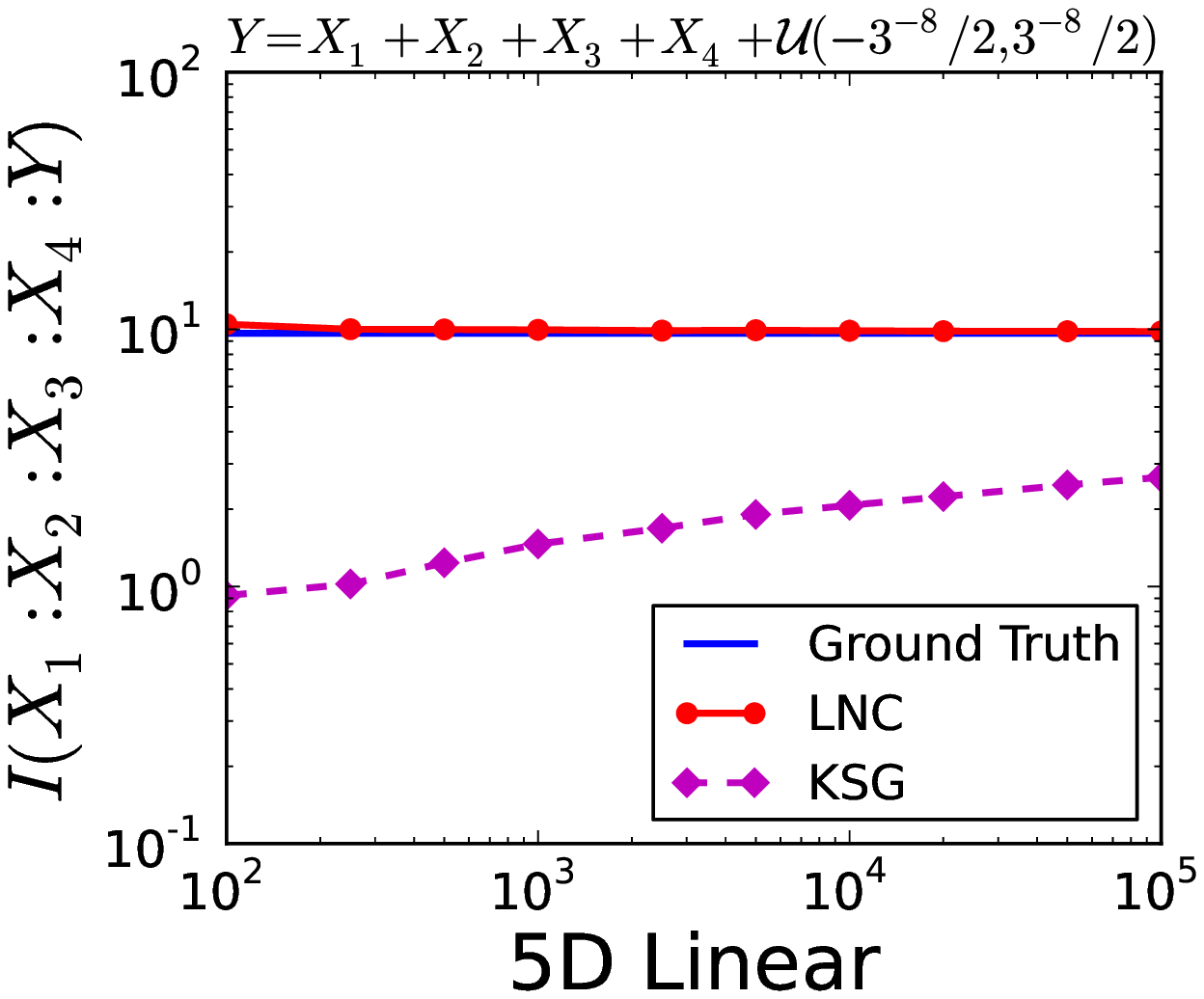}
\endminipage
\minipage{0.22\textwidth}
\includegraphics[width=\linewidth]{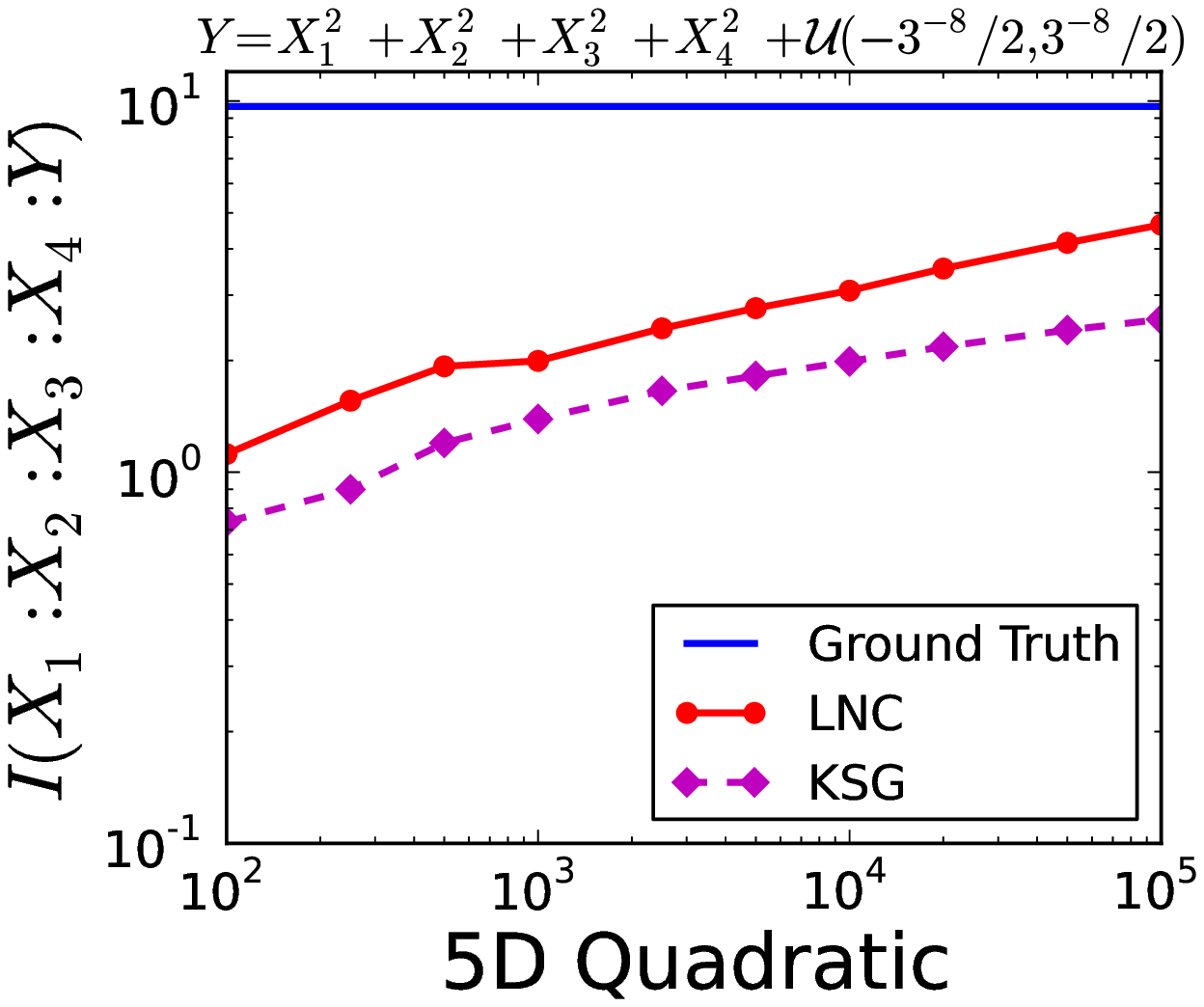}
\endminipage

\caption{Estimated MI using both KSG and LNC estimators in the number of samples ($k=5$ and $\alpha_{k,d}=0.37$ for 2D examples; $k=8$ and $\alpha_{k,d}=0.12$ for 5D examples)}
\label{fig:converge_rate}
\end{figure}

\subsection{Experiments with real-world data}
\subsubsection{Ranking Relationship Strength}
We evaluate the proposed estimator on the WHO dataset which has $357$ variables describing various socio-economic, political, and health indicators for different countries~\footnote{WHO dataset is publicly available at \url{http://www.exploredata.net/Downloads}}. We calculate the mutual information between pairs of variables which have at least 150 samples. Next, we rank the pairs based on their estimated mutual information and choose the top 150 pairs with highest mutual information. For these top 150 pairs, We randomly select a fraction $\rho$ of samples for each pair, hide the rest samples and then recalculate the mutual information. We want to see how mutual information-based rank changes by giving different amount of less data, i.e., varying $\rho$. A good mutual information estimator should give a similar rank using less data as using the full data. We compare our LNC estimator to KSG estimator. Rank similarities are calculated using the standard Spearman's rank correlation coefficient described in~\citep{spearman}. Fig~\ref{fig:WHO_rank} shows the results. We can see that LNC estimator outperforms KSG estimator, especially when the missing data approaches 90\%, Spearman correlation drops to 0.4 for KSG estimator, while our LNC estimator still has a relatively high score of 0.7.

\begin{figure}[htbp]
\centering
\includegraphics[width=0.682\linewidth]{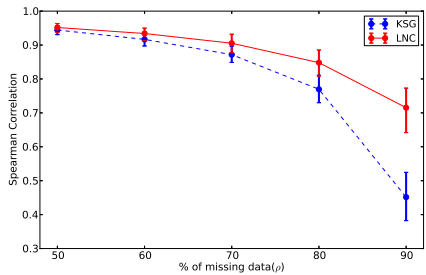}
\caption{Spearman correlation coefficient between the original MI rank and the rank after hiding some percentage of data by KSG and LNC estimator respectively. The 95\% confidence bars are obtained by repeating the experiment for 200 times.}

\label{fig:WHO_rank}
\end{figure}

\subsubsection{Finding interesting triplets}
We also use our estimator to find strong {\em multivariate} relationships in the WHO data set. Specifically, we search for {\em synergistic} triplets $(X,Y,Z)$, where one of the variables, say $Z$, can be predicted by knowing both $X$ and $Y$ simultaneously, but not by using either variable separately. In other words, we search for triplets 
$(X,Y,Z)$ such that the pair-wise mutual information between the pairs $I(X:Y)$, $I(X:Z)$ and $I(Y:Z)$ are low,  but the multi-information $I(X:Y:Z)$ is relatively high. We rank relationships using the following synergy score:\footnote{Another measure of synergy is given by the so called ``interaction information": $I\left( {X:Y} \right) + I\left( {Y:Z} \right) + I\left( {Z:X} \right) - I\left( {X:Y:Z} \right)$.} $SS = {I\left( {X:Y:Z} \right)}/{\max \left\{ {I\left( {X:Y} \right),I\left( {Y:Z} \right),I\left( {Z:X} \right)} \right\}}$.
%$$
%SS = \frac{{I\left( {X:Y:Z} \right)}}{{\max \left\{ {I\left( {X:Y} \right),I\left( {Y:Z} \right),I\left( {Z:X} \right)} \right\}}}. 
%$$
%While detecting strong relationships between pairs of variables in large data set is important, the generalized mutual information allows us to discover multivariate relationships. One fundamental concept of multivariate relationship is \textit{synergy}. The "synergistic information" can be understood that a set of random variables ${X_1,X_2,...,X_N}$ interacting together to predict a single variable $Y$ while none of them can predict $Y$ alone.

%In our experiments, WHO\footnote{WHO data set is publicly available at $http://www.exploredata.net/Downloads$}  data set collected by~\citep{reshef} is used and the task is to find the synergistic behavior among three variables in this data set. We calculate the mutual information between each of the two variables, each of the three variables (357 variables, 49,618 valid pairs, 4.6 million valid triplets) 

%We calculated mutual information for by using both KSG and LNC estimators. Focusing on synergistic behavior, we want to find the triplets $(X,Y,Z)$ such that the pair-wised mutual information $I(X:Y)$, $I(X:Z)$ and $I(Y:Z)$ are low but the multi-information $I(X:Y:Z)$ is high. To achieve this, we construct the normalized multi-information(NMI) score, defined as 
%\be
%NMI = \frac{{I\left( {X:Y:Z} \right)}}{{\max \left\{ {I\left( {X:Y} \right),I\left( {Y:Z} \right),I\left( {Z:X} \right)} \right\}}}
%\ee 

We select the triplets that have synergy score above a certain threshold. Figure~\ref{fig:synergy} shows two synergistic relationships detected by $LNC$ but not by $KSG$. In these examples, both $KSG$ and $LNC$ estimators yield low mutual information for the pairs $(X,Y)$, $(Y,Z)$ and $(X,Z)$. However, in contrast to KSG, our estimator yields a relatively high score for multi-information among the three variables.  
\begin{figure}[!htb]
\centering
\subfigure[]{
    \includegraphics[width=0.74\linewidth]{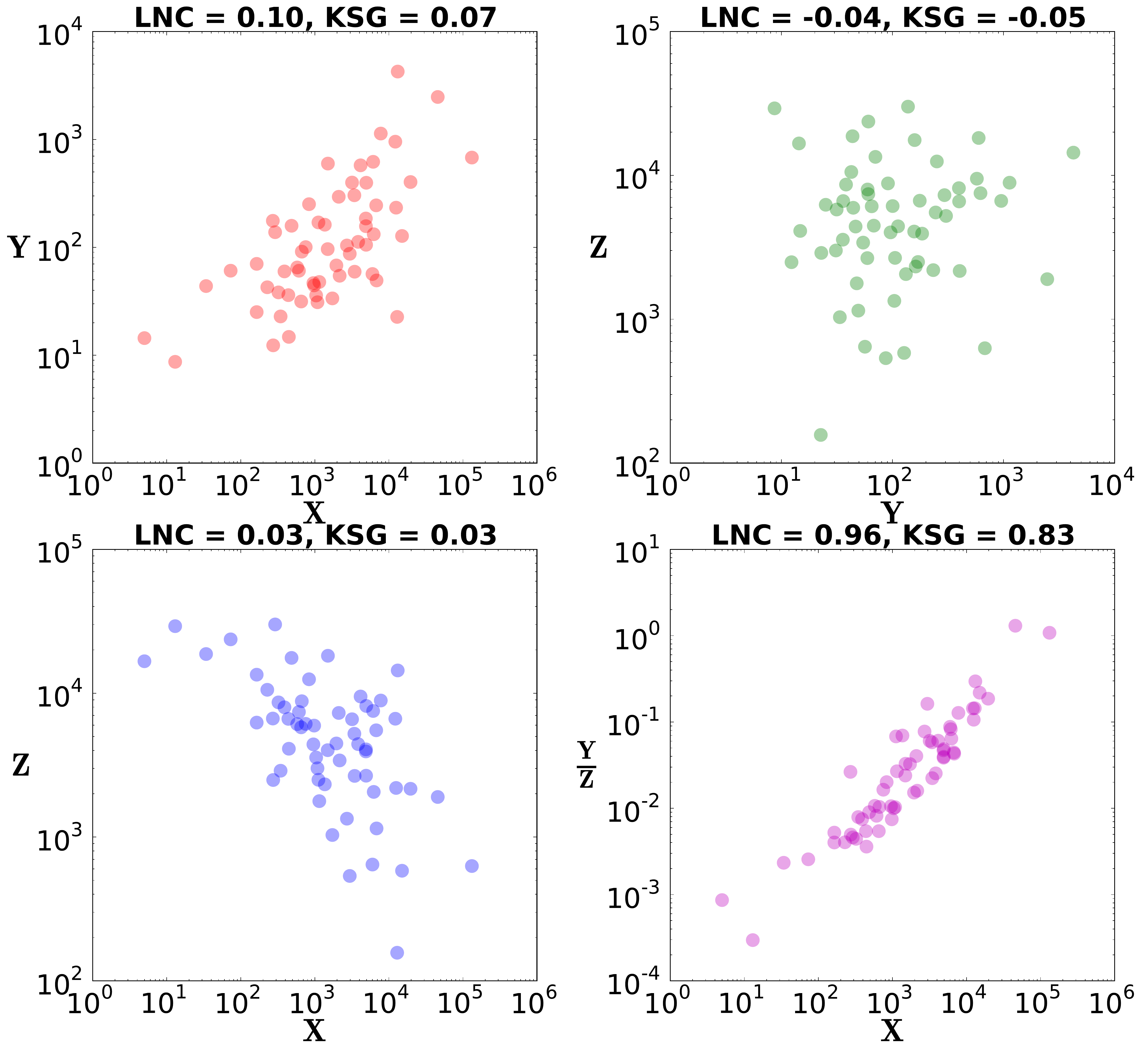} 

	\label{fig1a}
    } 

    \subfigure[]{
    \includegraphics[width=0.74\linewidth]{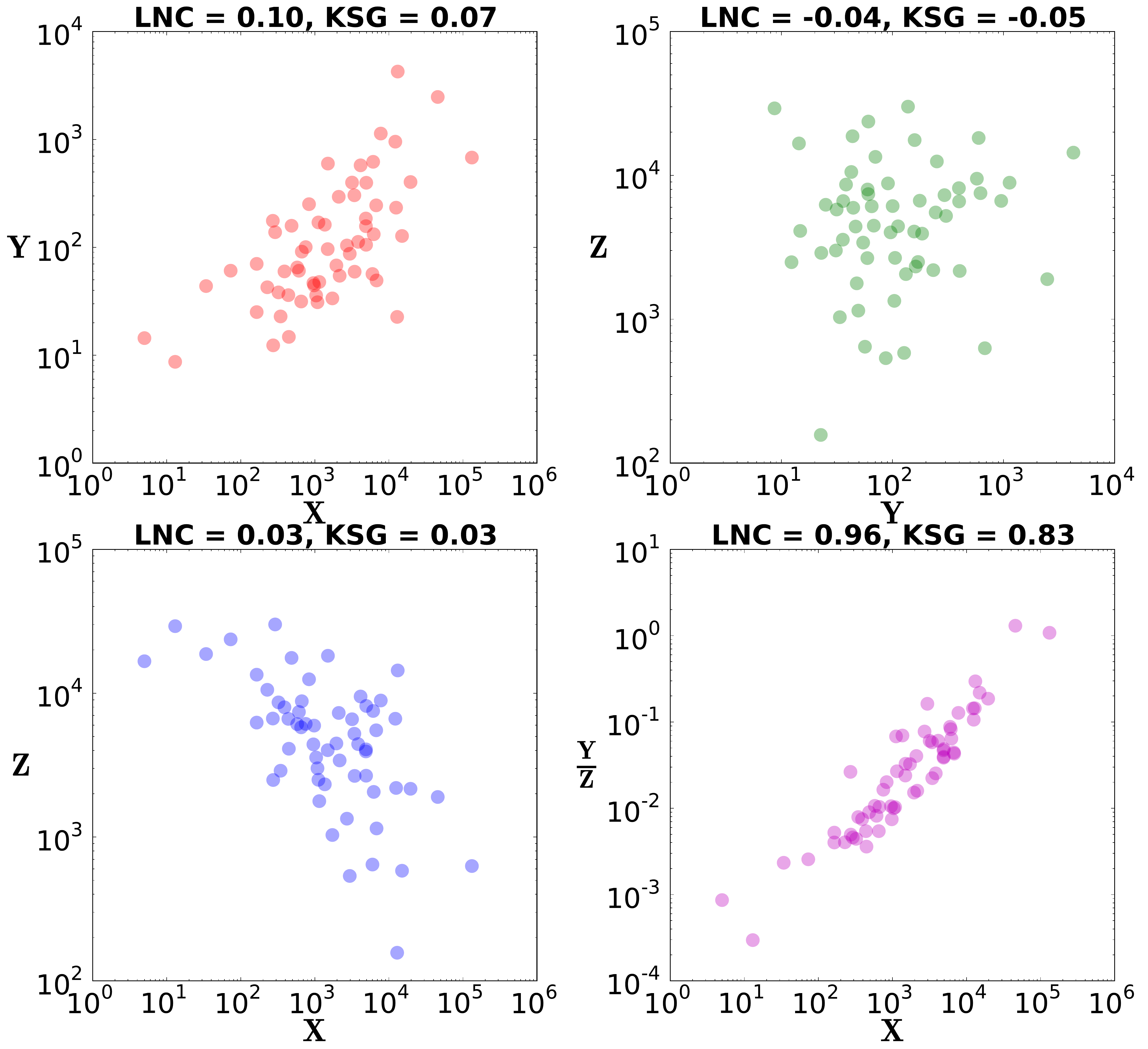}\label{fig1b} 
    }
    \caption{ Two examples  of synergistic triplets: $\hat I_{KSG}(X:Y:Z) = 0.14$ and $\hat I_{LNC}(X:Y:Z) = 0.95$ for the first example; $\hat I_{KSG}(X:Y:Z) = 0.05$ and $\hat I_{LNC}(X:Y:Z) = 0.7$ for the second example}
    \label{fig:synergy}
\end{figure}

For the first relationship, the synergistic behavior can be explained by noting that the ratio of the Total Energy Generation ($Y$) to Electricity Generation per Person  ($Z$) essentially yields the size of the population, which is highly predictive of the Number of Female Cervical Cancer cases  ($X$). While this example might seem somewhat trivial, it illustrates the ability of our method to extract synergistic relationships automatically without any additional assumptions and/or data preprocessing. %Furthermore, our approach produces non-trivial triplets, such as shown in the second example. 
%Here we observe strong synergistic interactions between 

In the second example, LNC predicts a strong synergistic interaction between Total Cell Phones ($Y$),  Number of Female Cervical Cancer cases ($Z$), and rate of Tuberculosis deaths ($X$). Since the variable $Z$ (the number of female cervical cancer cases) grows with the total population,  $\frac{Y}{Z}$ is proportional to the average number of cell phones per person. The last plot indicates that a higher number of cell phones per person are predictive of lower tuberculosis death rate. One possible explanation for this correlation is some common underlying cause (e.g., overall economic development). Another intriguing possibility is that this finding reflects recent efforts to use mobile technology in TB control.\footnote{See {\em Stop TB Partnership}, \url{http://www.stoptb.org}.}

\section{Related Work}\label{sec:related}

% survey~\citep{MI_survey}
\paragraph{MI estimators} There has been a significant amount of recent work on estimating entropic measures such as divergences and mutual information from samples (see this survey~\citep{MI_survey} for an exhaustive list). \citet{khan} compared different MI estimators for varying sample sizes and noise intensity, and reported that for small samples, the KSG estimator was the best choice overall for relatively low noise intensities, while KDE performed better at higher noise intensities. Other approaches include estimators based on Generalized Nearest-Neighbor Graphs~\citep{GNNG}, minimum spanning trees~\citep{spanning}, maximum likelihood density ratio estimation~\citep{densityratio}, and ensemble methods~\citep{Sricharan_ensemble,moon_ensemble_divergence}. In particular, the latter approach works by taking a weighted average of  simple density plug-in estimates such as kNN or KDE. However, it is upper-bounded by the largest value among its simple density estimates. Therefore, this method would still underestimate the mutual information when it goes larger as discussed before.
%, and more general dependence measures~\citep{sohan}. 
%survey paper ~\citep{Paninski2
%\note{do we know if this estimator suffers from the same problem as KSG?} 
% I think it does, -GV.

It has been recognized that kNN-based entropic estimators underestimate probability density at the sample points that are close to the boundary of support~\citep{boundarykNN}. 
%\citet{rahvar_boundary} suggested a simple correction method that uses smaller rectangles instead of hypercubes, which allows for more accurate density estimation at those points. However, their method only works at the global boundaries, not locally, and it fails whenever the boundaries are not parallel to the axes. 
\citet{Sricharan_BPI} proposed an bipartite plug-in(BPI) estimator for non-linear density functionals that  extrapolates the density estimates at {\em interior} points that are close to the {\em boundary} points in order to compensate the boundary bias. However, this method requires to identify boundary points and interior points which becomes difficult to distinguish as mutual information gets large that almost all the points are close to the boundary. 
\citet{singh_generalized_2014} used a "mirror image" kernel density estimator to escape the boundary effect, but their estimator relies on the knowledge of the support of the densities, and assumes that the boundary are parallel to the axes. %All these correction methods  are problematic for strongly dependent variables, where most of the points are close to the boundary. 

%Finally, Sricharan et al~\citep{Sricharan_ensemble}, Moon and Hero~\citep{moon_ensemble_divergence} derived fast rates for entropy and $f$-divergences estimation using an ensemble estimator. For this method they calculate the weighted average of simple density plug-in estimates such as kNN or KDE. However, the estimate is upper-bounded by the largest value among its simple density estimates. Therefore, this method would still underestimate the mutual information when it goes larger as discussed before.
%Significantly more work has focused on discrete entropy estimation including analytic~\citep{} and boostrap~\citep{} corrections to the plug-in estimators, and Bayesian methods~\citep{}.
%Discrete estimators \citep{versteeg2012www,dedeo}, Other dependence measures\citep{hhg,dcor,sohan}
\paragraph{Mutual Information and Equitability}
\citet{reshef} introduced a property they called ``suitability'' for a measure of correlation. If two variables are related by a functional form with some noise, equitable measures should reflect the magnitude of the noise while being insensitive to the form of the functional relationship. They used this notion to justify a new correlation measure called MIC. Based on comparisons with MI using the KSG estimator, they concluded that MIC is ``more equitable'' for comparing relationship strengths. While several problems~\citep{simoncomment,gorfinecomment} and alternatives~\citep{hhg,dcor} were pointed out, Kinney and Atwal (KA) were the first to point out that MIC's apparent superiority to MI was actually due to flaws in estimation~\citep{kinney}. A more careful definition of equitability led KA to the conclusion that MI is actually more equitable than MIC.  KA suggest that the poor performance of the KSG estimator that led to Reshef et. al.'s mistaken conclusion could be improved by using more samples for estimation. However, here we showed that the number of samples required for KSG is prohibitively large, but that this difficulty can be overcome by using an improved MI estimator. 

%Several authors pointed out that MIC itself fails to be equitable in many situations~\citep{simoncomment,gorfinecomment,kinney} and that other measures are often preferable~\citep{hhg,dcor}.
%Follow up paper of Reshef \citep{reshef2013}

%Kinney and Atwal~\citep{kinney} criticized the equitability notion in ~\citep{reshef} and argued that no nontrivial function can satisfy it. Instead, they suggested an intuitive condition they called  {\em self-equitability}
%and showed that it is satisfied by mutual information. They also remarked (correctly) that the poor performance of MI reported in \citep{reshef} was due to a flawed estimator, although they attributed this shortcoming to non-optimal choice of $k$. As we have demonstrated here, however, the cause of the problem with KSG is the assumption of uniformity of density in the max-norm box, rather than sub-optimal choice of $k$. While their solution for overcoming the limitation of KSG was to add more points, we showed that this quickly becomes impractical for the KSG estimator, while our method works well even with limited samples.  

%\citep{wsdm2013} \citep{icml2014}
%Code implemented~\citep{npeet}
%{kNNdensity} Density estimation using only the k-nn graph requires much larger neighborhoods.

\section{Conclusion}\label{sec:conclusion}
The problem of deciding whether or not two variables are independent is a historically significant endeavor. 
%Statistics has largely focused on the classic problem of determining whether it is possible to exclude the null hypothesis that two variables are independent. 
In that context, research on mutual information estimation has been geared towards distinguishing weak dependence from independence. However, modern data mining presents us with problems requiring a totally different perspective. It is not unusual to have thousands of variables which could have millions of potential relationships. We have insufficient resources to examine each potential relationship so we need an assumption-free way to pick out only the most promising relationships for further study. Many applications have this flavor including the health indicator data considered above as well as gene expression microarray data, human behavior data, economic indicators, and sensor data, to name a few. 

How can we select the most interesting relationships to study? Classic correlation measures like the Pearson coefficient bias the results towards linear variables. Mutual information gives a clear and general basis for comparing the strength of otherwise dissimilar variables and relationships. While non-parametric mutual information estimators exist, we showed that strong relationships require exponentially many samples to accurately measure using some of these techniques. 

We introduced a non-parametric mutual information estimator that can measure the strength of nonlinear relationships even with small sample sizes. We have incorporated these novel estimators into an open source entropy estimation toolbox~\footnote{\url{https://github.com/BiuBiuBiLL/MIE}}.% We demonstrated that our technique works well on both synthetic and real-world data. 
As the amount and variety of available data grows, general methods for identifying strong relationships will become increasingly necessary. We hope that the developments suggested here will help to address this need. 

\subsubsection*{Acknowledgements}
This research was supported in part by DARPA grant No. W911NF--12--1--0034.

\bibliographystyle{plainnat}

\appendix

\clearpage
\counterwithin{figure}{section}
\counterwithin{equation}{section}
\counterwithin{algorithm}{section}

\section*{Supplementary Material for ``Efficient Estimation of Mutual Information for Strongly Dependent Variables''}

\section{Proof of Theorem~\ref{theo:kNN}}
\label{sec:derive_kNN}
Notice that for a fixed sample point $\vx^\ii$, its $k$-nearest-neighbor distance ${{r_k}\left( {{{\mathbf{x}}^\ii}} \right)}$ is always equal to or larger than the $k$-nearest-neighbor distance of at the same point $\vx^\ii$ projected into a sub-dimension $j$, i.e., for any $i,j$, we have
\be \label{eq:kNN_dist}
{r_k}\left( {{{\mathbf{x}}^\ii}} \right) \geqslant {r_k}\left( {{{\mathbf{x}}^\ii_j}} \right)
\ee

Using Eq.~\ref{eq:kNN_dist}, we get the upper bound of ${\widehat I_{kNN,k}}\left( {\mathbf{x}} \right)$ as follows:
\BEA \label{eq:kNN_eq1}
  {\widehat I_{kNN,k}}\left( {\mathbf{x}} \right) &=& \widehat I{'_{kNN,k}}\left( {\mathbf{x}} \right) - \left( {d - 1} \right){\gamma _k} \nonumber \\
   &=& \frac{1}{N}\sum\limits_{i = 1}^N {\log \frac{{{{\widehat p}_k}\left( {{{\mathbf{x}}^\ii}} \right)}}{{\prod\limits_{j = 1}^d {{{\widehat p}_k}\left( {{{\mathbf{x}}^\ii_j}} \right)} }} - \left( {d - 1} \right){\gamma _k}}  \nonumber \\
   &=& \frac{1}{N}\sum\limits_{i = 1}^N {\log \frac{{\frac{k}{{N - 1}}\frac{{\Gamma \left( {d/2} \right) + 1}}{{{\pi ^{d/2}}}}{r_k}{{\left( {{{\mathbf{x}}^\ii}} \right)}^{ - d}}}}{{\prod\limits_{j = 1}^d {\frac{k}{{N - 1}}\frac{{\Gamma \left( {1/2} \right) + 1}}{{{\pi ^{1/2}}}}{r_k}{{\left( {{{\mathbf{x}}^\ii_j}} \right)}^{ - 1}}} }}}\nonumber \\
  &~&- \left( {d - 1} \right){\gamma _k} \nonumber\\
   &\le& \left( {d - 1} \right)\log \left( {\frac{{N - 1}}{k}} \right) + \log \frac{{\Gamma \left( {d/2} \right) + 1}}{{{{\left( {\Gamma \left( {1/2} \right) + 1} \right)}^d}}} \nonumber \\
   &~&- \left( {d - 1} \right)\left( {\psi \left( k \right) - \log k} \right)  \nonumber \\
 &\le& \left( {d - 1} \right)\log \left( {\frac{{N - 1}}{k}} \right) + \log \frac{{\Gamma \left( {d/2} \right) + 1}}{{{{\left( {\Gamma \left( {1/2} \right) + 1} \right)}^d}}} \nonumber \\
 &~&- \left( {d - 1} \right)\left( {\psi \left( 1 \right) - \log 1} \right) \nonumber \\
\EEA
The last inequality is obtained by noticing that $\psi(k)-\log(k)$ is a monotonous decreasing function.

Also, we have,
\BEA \label{eq:kNN_eq2}
  \log \frac{{\Gamma \left( {d/2} \right) + 1}}{{{{\left( {\Gamma \left( {1/2} \right) + 1} \right)}^d}}} &=& \log \left( {\Gamma \left( {d/2} \right) + 1} \right) - d\log \left( {\Gamma \left( {d/2} \right) + 1} \right) \nonumber \\
   &<& \log \left( {\sqrt {2\pi } {{\left( {\frac{{d/2 + 1/2}}{e}} \right)}^{d/2 + 1/2}}} \right) \nonumber \\
   &~&- d\log \left( {{\pi ^{\frac{1}{2}}} + 1} \right) \nonumber \\
   &=& O\left( {d\log d} \right) \nonumber
\EEA

The inequality above is obtained by using the bound of gamma function that,
\[\Gamma \left( {x + 1} \right) < \sqrt {2\pi } {\left( {\frac{{x + 1/2}}{e}} \right)^{x + 1/2}}\]
Therefore, reconsidering \ref{eq:kNN_eq1}, we get the following inequality for ${\widehat I_{kNN,k}}\left( {\mathbf{x}} \right)$:
\be
{\widehat I_{kNN,k}}\left( {\mathbf{x}} \right) &\le& \left( {d - 1} \right)\log \left( {\frac{{N - 1}}{k}} \right) + O\left( {d\log d} \right) \nonumber \\
&\le& \left( {d - 1} \right)\log \left( {N - 1} \right) + O\left( {d\log d} \right)
\ee

Requiring that $|{\widehat I_{kNN,k}}\left( {\mathbf{x}} \right) - I(\vx)| \le \varepsilon$, we obtain,
\be
N \ge C \exp \left( {\frac{{I\left( {\mathbf{x}} \right) - \epsilon }}{{d - 1}}} \right) + 1
\ee 
where $C$ is a constant which scales like $O(\frac{1}{d})$.

\section{Derivation of Eq.~\ref{eq:lnc}}
\label{sec:derive_lnc}
The naive kNN or KSG estimator can be written as:
\be\label{eq:ksg_eq_2}
{\widehat I_{k}}\left( {\mathbf{x}} \right) = \frac{1}{N}\sum\limits_{i = 1}^N {\log \frac{{\frac{{{P}\left( {{{\mathbf{x}}^\ii}} \right)}}{{V\left( i \right)}}}}{{\prod\limits_{j = 1}^d {\frac{{{P}\left( {{{\mathbf{x}}^\ii_j}} \right)}}{{{V_j}\left( i \right)}}} }}} 
\ee
where $P(\vx^\ii)$ is the probability mass around the $k$-nearest-neighborhood at $\vx^\ii$ and $P(\vx^\ii_j)$ is the probability mass around the $k$-nearest-neighborhood (or $n_{x_j}(i)$-nearest-neighborhood for KSG) at $\vx^\ii$ projected into $j$-th dimension.  Also, $V(i)$ and $V_j(i)$ denote the volume of the kNN ball(or hype-rectangle in KSG) in the joint space and projected subspaces respectively.

Now our local nonuniform correction method replaces the volume $V(i)$ in Eq.~\ref{eq:ksg_eq_2} with the corrected volume $\overline{V}(i)$, thus, our estimator is obtained as follows:
\BEA
  {\widehat I_{LNC,k}}\left( {\mathbf{x}} \right) &=& \frac{1}{N}\sum\limits_{i = 1}^N {\log \frac{{\frac{{{P}\left( {{{\mathbf{x}}^\ii}} \right)}}{{\overline V \left( i \right)}}}}{{\prod\limits_{j = 1}^d {\frac{{{P}\left( {{{\mathbf{x}}^\ii_j}} \right)}}{{{V_j}\left( i \right)}}} }}}  \nonumber \\
   &=& \frac{1}{N}\sum\limits_{i = 1}^N {\log \frac{{\frac{{{P}\left( {{{\mathbf{x}}^\ii}} \right)}}{{V\left( i \right)}} \times \frac{{V\left( i \right)}}{{\overline V \left( i \right)}}}}{{\prod\limits_{j = 1}^d {\frac{{{P}\left( {{{\mathbf{x}}^\ii_j}} \right)}}{{{V_j}\left( i \right)}}} }}} \nonumber \\
    &=& {\widehat I_{k}}\left( {\mathbf{x}} \right) - \frac{1}{N}\sum\limits_{i = 1}^N {\log \frac{{\overline V \left( i \right)}}{{V\left( i \right)}}} \nonumber \\ 
\EEA

\section{Empirical Evaluation for $\alpha_{k,d}$}
\label{sec:derive_alpha}

Suppose we have a uniform distribution on the $d$ dimensional (hyper)rectangle with volume $V$.  We sample $k$ points from this uniform distribution. We perform PCA using these $k$ points to get a new basis. After rotating into this new basis, we find the volume, $\bar V$, of the smallest rectilinear rectangle containing the points. 
By chance, we will typically find $\bar V < V$, even though the distribution is uniform. This will lead to us to (incorrectly) apply a local non-uniformity correction. Instead, we set a threshold $\alpha_{k,d}$ and if $\bar V/V$ is above the threshold, we assume that the distribution is locally uniform. Setting $\alpha$ involves a trade-off. If it is set too high, we will incorrectly conclude there is local non-uniformity and therefore over-estimate the mutual information. If we set $\alpha$ too low, we will lose statistical power for ``medium-strength'' relationships (though very strong relationships will still lead to values of $\bar V/V$ smaller than $\alpha$). 

In practice, we determine the correct value of $\alpha_{k,d}$ empirically. 
We look at the probability distribution of $\bar V/V$ that occurs when the true distribution is uniform. 
We set $\alpha$ conservatively so that when the true distribution is uniform, our criteria rejects this hypothesis with small probability, $\epsilon$.  
Specifically, we do a number of trials, $N$, and set $\hat \alpha_{k,d}$ such that $\sum\limits_{i = 1}^N {{\bf{I}}\left( {\frac{{\bar {{V_i}}}}{{{V_i}}} < {{\hat \alpha }_{k,d}}} \right)} /N < \epsilon $ where $\epsilon$ is a relatively small value. In practice, we chose $\epsilon = 5\times10^{-3}$ and $N = 5\times{10^5}$.  The following algorithm describes this procedure:

\begin{algorithm}[h]
\caption{\textbf{Estimating $\alpha_{k,d}$ for LNC}}
\begin{algorithmic}
\label{alg:alpha}
\State{\textbf{Input:} parameter $d$ (dimension), $k$ (nearest neighbor), $N$, $\epsilon$}
\State{\textbf{Output:} $\hat \alpha_{k,d}$}

\State{set array $\bf{a}$ to be NULL}
\Repeat
	\State{Randomly choose a uniform distribution supported on $d$ dimensional (hyper) rectangle, denote its volume to be $V$}
	\State{Draw $k$ points from this uniform distribution, get the correcting volume $\bar V$ after doing PCA}
	\State{add the ratio $\frac{\bar V}{V}$ to array $\bf{a}$}
\Until above procedure repeated $N$ times

\State{$\hat \alpha_{k,d} \leftarrow \left\lceil\epsilon N\right\rceil th$ smallest number in $\bf{a}$}
\end{algorithmic}
\end{algorithm}

\begin{figure}[ht]
	\centering
	\includegraphics[width=1.0\linewidth]{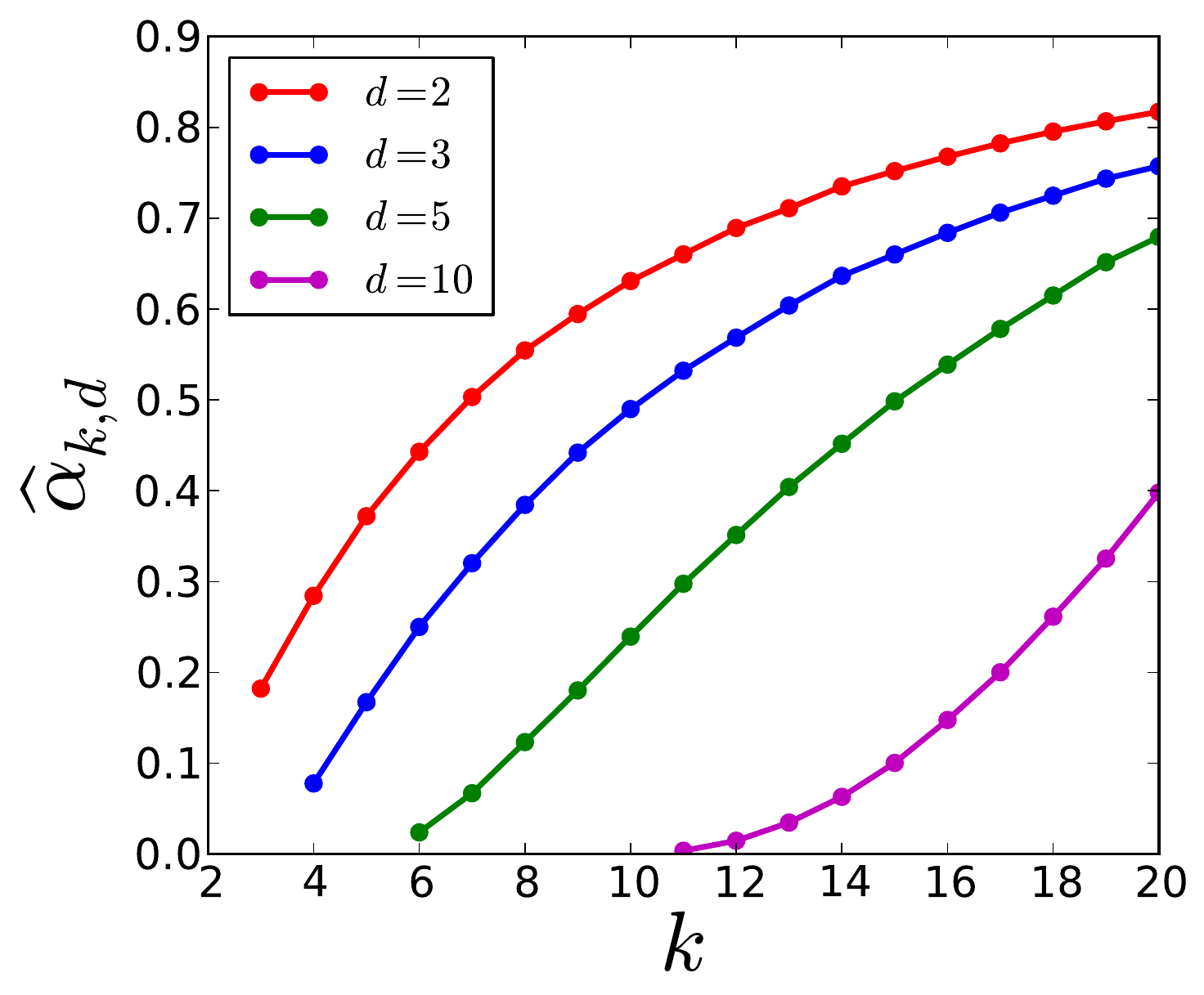}
	\caption{$\widehat \alpha_{k,d}$ as a function of $k$. $k$ ranges over $[d,20]$ for each dimension $d$.}
	\label{fig:alpha}
\end{figure}

Figure~\ref{fig:alpha} shows empirical value of $\hat \alpha_{k,d}$ for different $(k,d)$ pairs. We can see that for a fixed dimension $d$, $\hat \alpha_{k,d}$ grows as $k$ increases, meaning that $\bar V$ must be closer to $V$ to accept the null hypothesis of \textit{uniformity}.
We also find that $\hat \alpha_{k,d}$ decreases as the dimension $d$ increases, indicating that for a fixed $k$, $\bar V$ becomes much smaller than $V$ when points are drawn from a uniform distribution in higher dimensions.

\section{More Functional Relationship Tests in Two Dimensions}\label{sec:more}
We have tested together twenty-one functional relationships described in ~\citet{reshef, kinney}, we show six of them in Section~\ref{sec:results}. The complete results are shown in Figure~\ref{fig:all_functions}. Detailed description of the functions can be found in Table S1 of Supporting Information in~\citet{kinney}. 

\begin{sidewaysfigure*}[ht]
    \includegraphics[width=1.0 \textwidth]{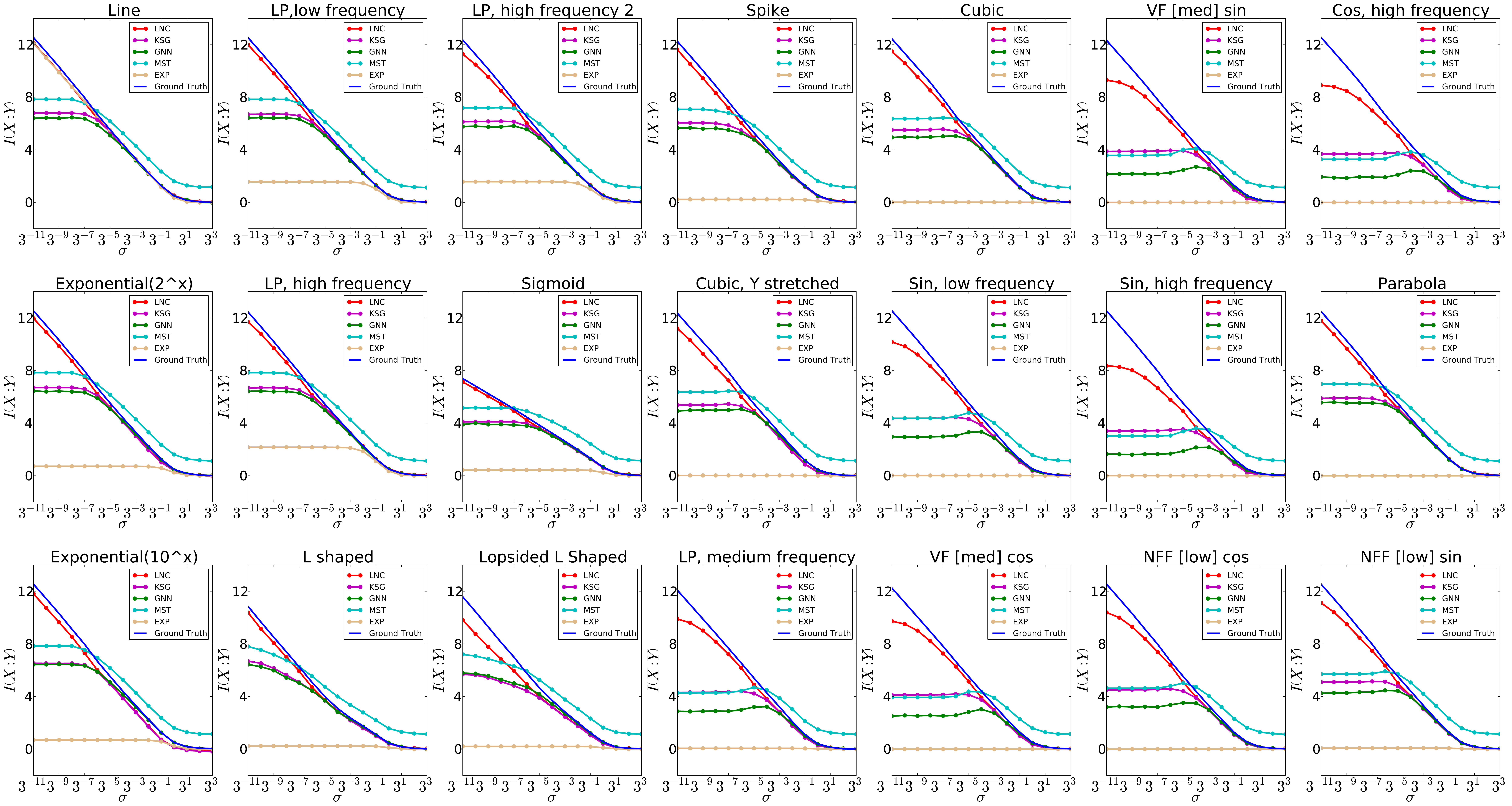}
    \caption{Mutual Information tests of $LNC$, $KSG$, $GNN$, $MST$, $EXP$ estimators. Twenty-one functional relationships with different noise intensities are tested. Noise has the form $U[-\sigma/2,\sigma/2]$where $\sigma$ varies(as shown in X axis of the plots). For KSG, GNN and LNC estimators, nearest neighbor parameter $k=5$. We are using $N = 5,000$ data points for each noisy functional relationship.}
    \label{fig:all_functions}
\end{sidewaysfigure*}

\end{document}